\documentclass{article}

\usepackage[numbers]{natbib}
\usepackage{graphicx}
\usepackage{amssymb, amsmath, latexsym, amsfonts, amsthm}
\usepackage{authblk}
\DeclareMathOperator{\argmax}{\arg\,\max}
\usepackage{color}
\newtheorem{theorem}{Theorem}
\newtheorem{lemma}{Lemma}
\newtheorem{definition}{Definition}
\newtheorem{corollary}{Corollary}
\begin{document}
	
\title{Hotelling-Downs Model with Limited Attraction}
\author[1]{Weiran Shen}
\author[2]{Zihe Wang}
\affil[1]{IIIS, Tsinghua University}
\affil[2]{ITCS, Shanghai University of Finance and Economics}
%\author{
%	Weiran Shen \\
%	IIIS, Tsinghua University \\
%	\texttt{emersonswr@gmail.com}
%	\and
%	Zihe Wang \\
%	ITCS, Shanghai University of Finance and Economics\\
%	\texttt{wzh5858588@163.com}
%}
\date{}

\maketitle
\begin{abstract}
In this paper we study variations of the standard Hotelling-Downs model of spatial competition, where each agent attracts the clients in a restricted neighborhood, each client randomly picks one attractive agent for service. 

Two utility functions for agents are considered: support utility and winner utility.
We generalize the results by Feldman et al. to the case where the clients are distributed arbitrarily.
In the support utility setting, we show that a pure Nash equilibrium always exists by formulating the game as a potential game.
In the winner utility setting, we show that there exists a Nash equilibrium in two cases: when there are at most 3 agents and when the size of attraction area is at least half of the entire space.
We also consider the price of anarchy and the fairness of equilibria and give tight bounds on these criteria.
%Then we show the price of anarchy is at least 0.5 in the support utility maximization setting. We also consider other criteria: the support utility comparison among agents and how to reach a Nash equilibrium.
\end{abstract}

%\keywords{Hotelling-Downs Model, pure Nash equilibrium, price of anarchy}

\section{Introduction}

%Ever since the seminal works by Hotellings~\cite{hotelling1929stability} and Downs~\cite{downs1957economic}, the Hotelling-Downs model has become the standard model for many problems, ranging from determining the standpoint of a election candidate to choosing locations for commercial facilities. In the standard Hotelling-Downs model, two firms compete by choosing locations in an interval to set up shops. 
%The customers are distributed along the interval and always go to the closer shop. The firms' objective is to attract as many customers as possible. Since moving towards the competitor's location always attracts more customers, both firms choose the median point in the unique stable equilibrium, and thus attract equal amount of customer. This to some extent explains the phenomenon that the candidates' standpoints in a political election are often close to one another.
Ever since the seminal works by Hotellings~\cite{hotelling1929stability} and Downs~\cite{downs1957economic}, the Hotelling-Downs model has been applied to many problems, ranging from determining the standpoint of a election candidate to choosing locations for commercial facilities~\cite{osborne1995spatial,brenner2011location,wunder2012framework}. 
In the model, two firms choose shop locations in a line market.
Customers are distributed along the line. 
Assume the products of the firms are uniform, the customers always go to the closer shop. 
Hence, one firm can always attract more customers by moving towards the competitor's location.
Therefore, both firms choose the median point in the unique stable equilibrium, attracting an equal number of customers.
This also sheds light upon the phenomenon that the candidates' standpoints in a political election are often close.

%There are a large body of literatures following Hotelling-Downs model~\cite{osborne1995spatial,brenner2011location,wunder2012framework}. 
%In political election scenarios, to explain why the standpoints of the candidates differ in reality, \cite{wittman1990spatial,duggan2005electoral} assume that every candidate has an ideal location in mind and cares about how close the winner's location is to the ideal location.
%\cite{sengupta2008hotelling,brusco2012hotelling,xefteris2014mixed} study the model with runoff voting, where the voting takes place in multi rounds and only a subset of candidates from the previous round remains in the next round.
%In commercial facility location problems, more complex models are proposed to deal with other realistic issues. For example, when making decisions, customers consider the transportation cost caused by getting service.
%%It is also natural to introduce the transportation cost of the customer, firms can also set a price besides choosing a location.
%Such competition that involves both location and pricing are considered by Hotelling's original model~\cite{hotelling1929stability}, and some other existing works~\cite{dasgupta1986existence,thisse1988strategic,zhang1995price,pinkse2002spatial,fang2016market}.

% There are a large body of literatures following Hotelling-Downs model~\cite{osborne1995spatial,brenner2011location,wunder2012framework}. 
In political election scenarios,  the assumption is taken to model the standpoints of the candidates in \cite{wittman1990spatial,duggan2005electoral} that  every candidate has an ideal location in mind and cares about how close the winner's location is to the ideal location.
Models with runoff voting are studied in \cite{sengupta2008hotelling,brusco2012hotelling,xefteris2014mixed}, 
where the voting takes place in multiple rounds and only a subset of candidates from the previous round enter the next round.
In commercial facility location problems, more complex models are proposed to address other issues. For example, when making decisions, customers consider the transportation cost caused by getting service.
%It is also natural to introduce the transportation cost of the customer, firms can also set a price besides choosing a location.
Such competition that involves both location and pricing are considered by Hotelling's original model~\cite{hotelling1929stability}, and some other existing works~\cite{dasgupta1986existence,thisse1988strategic,zhang1995price,pinkse2002spatial,fang2016market}.

%In this paper, we focus on the pure location game~\cite{ben2016multi}.
%%Pure location game in the original model sacrifices no Nash equilibrium.
%There does not exist Nash equilibrium in the original pure location game.
%\cite{eaton1975principle} first shows there is no Nash equilibrium when there are 3 agents in the one-dimensional space.
%Then \cite{shaked1975non} extends the non-existence to two-dimensional space.
%After that \cite{osborne1993candidate} shows Nash equilibrium fails to exist in a wide range of settings when there are more than 2 agents.
%However, a mixed Nash equilibrium guarantees to exist~\cite{dasgupta1986existence,shaked1982existence}.
%This result is not obvious noticing that the utility functions in these games are not continuous with the action. Generally, a mixed Nash equilibrium does not guarantee to exist in such games.

In this paper, we focus on the pure location game~\cite{ben2016multi}.
It sacrifices non-existence of Nash equilibrium in the original pure location game.
Eaton~\emph{et~al.}~\cite{eaton1975principle} first show there is no Nash equilibrium when there are $3$ agents in the one-dimensional space.
 Shaked~\emph{et~al.}~\cite{shaked1975non} extend the non-existence to two-dimensional space. 
Thereafter Osborne~\emph{et~al.}~\cite{osborne1993candidate} show the Nash equilibrium does not exist in a wide range of settings when there are more than $2$ agents.
However, a mixed Nash equilibrium is guaranteed to exist~\cite{dasgupta1986existence,shaked1982existence}. 
This result is not obvious considering that the utility functions in these games are not continuous with the action. 
Generally, a mixed Nash equilibrium is not guaranteed to exist in such games.

%As pointed out in \cite{feldman2016variations}, some assumptions in the original model need more discussions.
%In the original model, the customer chooses the nearest shop no matter how far it is from himself.
%In reality, the customer would give up shopping if the shop is too far away.
%Furthermore, in the original model, the customer chooses the nearest shop no matter where the other shops are.
%In fact if two shops are close enough, they should have equal chance to be selected.
%For example, let us consider the case of order online restaurants. 
%Every restaurant delivers food to the customer within some distance.
%As long as the customer is within the restaurant's delivery service area, the customer does not care about the exact distance of the restaurant any more.
%If the food provided by these available restaurants has no big difference, 
%then the customer will choose one restaurant randomly form these restaurants.

The original Hotelling-Downs model suffers from some problematic assumptions: customers always choose the nearest shop without considering the distance, contradicting to the fact that a shop is no more attractive to a customer if it is too far away.
Furthermore, customers choose the shop without considering competing shops, while in daily life, it is hard to say which shop attracts more customers if two shops are close enough with similar products. These issues are also discussed in \cite{feldman2016variations}.

%To capture such customer's behavior and to analyze shops' locations caused by such behaviors, 
%\cite{feldman2016variations} proposes a variant modifying the above two assumptions, which we call it Hotelling-Downs model with limited attraction.
%In this model, every firm (called agent hereafter) has attraction only on the customers (called clients hereafter) in some limited interval.
%We call this limited interval ``attraction interval".
%Only the clients in the agent's attraction interval may choose that agent.
%If an client lies in the intersection of several agents, the client chooses each agent with equal probability.

To address the above issues, we consider the Hotelling-Downs model with limited attraction, proposed in \cite{feldman2016variations}. In this model, all firms (called agents hereafter) only attract customers (called clients hereafter) in a limited distance, and if a client is attracted by multiple agents, the client picks one from those agents with equal chance.

%We consider two basic settings: support utility maximization and winner utility maximization (i.e., winner takes all setting in \cite{feldman2016variations}).
%Here the support utility of an agent denotes the total support he receives.
%In the former one, the unique goal of every agent is to maximize his own support utility. This corresponds to commercial scenario.
%In the latter one, winners are the agents with the largest support utility.
%1 utility divided by the number of winners is the winner utility.
%This corresponds to the political scenario, i.e, winner takes all.
%The difference in the former setting is that, even if an agent cannot win, he still seeks to maximize the support utility.
\iffalse
We analyze the Nash equilibria with two utility functions for agents:
support utility and winner utility(i.e., winner takes all setting in \cite{feldman2016variations}). 
In the support utility setting, agents focus on maximizing the number of its clients, modeling the commercial competitions.
While in the winner utility setting, the winner in the competition takes all the utility, which is typical for the political voting.
\fi

We analyze the Nash equilibria with two utility functions for agents:
support utility and winner utility(i.e., winner takes all setting in \cite{feldman2016variations}). 
In the support utility setting, agents focus on maximizing the number of its clients, modeling the commercial competitions.
While in winner utility setting, the winner in the competition takes all the utility, which is typical for the political voting.

%We study the same problem as \cite{feldman2016variations} and extend their results on uniform distribution to general distribution.
%First of all, we consider the existence of {\it pure} Nash equilibrium. 
%It is easy to construct a Nash equilibrium when the distribution is uniform and agent's have the same size of attraction interval( called width hereafter),
%since we can make use of symmetric property in the uniform case.
%However, the method does not work anymore undel general distribution.
%In winner utility maximization setting, the main difficulty is that there is no tool to use in such discontinuous game.
%As we have noticed, it is hard to guarantee a Nash equilibrium in the related literature when there are more than 2 agents.
%It really needs some effort to construct one.
\iffalse
We extend the results on uniform distribution in \cite{feldman2016variations} to arbitrary distributions.
First of all, we consider the existence of {\em pure} Nash equilibrium. 
When the distribution is uniform and agents have the same attracting distance( called width hereafter), the existence of Nash equilibra can shown by simply constructing one.
since we can make use of symmetry in the uniform case.
However, the method does not work any more under other distribution.
In winner utility maximization setting, the main difficulty is that there is no tool to use in such discontinuous game.
As we have noticed, it is hard to guarantee a Nash equilibrium in the related literature when there are more than 2 agents.
It really needs some efforts to construct one.
\fi

We extend the results on uniform distribution in \cite{feldman2016variations} to arbitrary distributions.
First of all, we consider the existence of {\em pure} Nash equilibrium. 
In support utility setting, when the distribution is uniform and agents have the same attracting distance(called width hereafter), the existence of Nash equilibra can shown by simply constructing one.
However, the method does not work any more under other distribution. We solve this problem by formulating it as a potential game.
In winner utility maximization setting, to our knowledge, there is no tool to guarantee the existence of Nash equilibria in games with such utility functions. However, we show that a pure Nash equilibrium does exist in some simple cases.

\iffalse
Secondly, we characterize the Nash equilibrium.
The difficulty is that we do not know what a Nash equilibrium exactly is and whether Nash equilibrium is unique or not.
We compare different agents' support utilities in a Nash equilibrium and check whether support utilities are approximately equal.
The utility decreases because of the agent's self-interest behavior, the price of anarchy is quantified to study the inefficiency.
Since the sum of agents' support utilities equal to the number of clients served, % that are provided services, 
the sum of agents' support utilities is exactly the participation rate.
\fi
Secondly, we study fairness and price of anarchy of Nash equilibria in support utility setting (both of which are straightforward in the winner utility setting). We give tight bounds on both criteria.

\subsection{Results and contributions}
\begin{itemize}
\item
In the support utility maximization setting, support utility is continuous. 
Applying Glicksberg's theorem~\cite{glicksberg1952further}, continuous game guarantees a mixed Nash equilibrium.
If we let the agents dynamically best respond to the other's locations (one agent each round), then the location profile converges to a Nash equilibrium.
Scrutinizing each agent's action, each improvement actually increases a potential function.
Thus the game admits a pure Nash equilibrium.

\item
In the winner utility maximization setting, the winner utility is not continuous any more. 
The game is somehow harder to analyze.
We restrict the problem to the case when the agents have the same width.
We prove that when there are at most 3 agents, there exists a pure Nash equilibrium. 
Three-agents case is very special in other variants~\cite{eaton1975principle,shaked1975non,shaked1982existence}.
The existence of Nash equilibrium complements their results.
If the agent's width is at least half of the client space, then there also exists a Nash equilibrium for any number of agents.

%\item In winner utility maximization setting, the loser could have a very low support utility compared to the winner, 
%due to the fact that the loser has no incentive to improve it.
%While in support-utility maximization setting, the support utility of each agent is bounded by that of any other agent.
%To be specific, we define the utility density to be  the support utility divided the width.
%The ratio between any two agent's utility density lies in $[0.5,2]$.

\item We study fairness and price of anarchy of Nash equilibria in support utility setting (both of which are straightforward in the winner utility setting). We show that the fairness criterion can be bounded by $\frac{1}{2\left\lceil w_M/w_m\right\rceil }$, where $w_M=\max\{w\}$ and $w_m=\min\{w\}$. We also prove that the price of anarchy at least $\frac{1}{2}$. Both bounds are tight.
%Considering price of anarchy, it could be very low in the winner utility maximization setting.
%But we have positive result in the support utility maximization setting, price of anarchy is at least 0.5.
%
%
%
%\item We discuss how to construct a Nash equilibrium efficiently. There is a case that might need arbitrarily many steps to reach the equilibrium.
\end{itemize}

The structure of this paper is as follows:
In Section~\ref{sect:models} , we describe the coined Hotelling-Downs model with limited attraction in \cite{feldman2016variations}.
In Section~\ref{sect:max}, we prove the existence of the Nash equilibrium in support-utility maximization setting.
In Section~\ref{sect:winner}, we construct a Nash equilibrium in winner utility maximization setting.
In Section~\ref{sect:fair}, the support utilities in Nash equilibrium are compared.
In Section~\ref{sect:poa}, the price of anarchy is given and we give an upper bound on the amount of  clients have not been served.

	\section{Hotelling-Downs model with limited attraction}
	\label{sect:models}
	
	We consider a one-dimensional location space, represented by interval $[0,1]$.
	A continuum of clients are distributed in the interval according to some density function $f(x)$. 
	Let $N=(1,2,\dots,n)$ denote a finite set of agents and each agent $i$ is associated with an attraction width $w_i$.
	Each agent chooses a location in $[0,1]$ and an attraction interval $R_i$ centered at the chosen location is formed.
	The agent thus obtain the support from the clients in his attraction interval. If a client is covered by multiple agents, the client simply randomly choose one, i.e. the support of this client is equally divided among these agents in expectation.
	Assume that agent $i$ chooses location $x_i$ then the attraction interval $R_i$ would be $[x_i-\frac{w_i}{2},x_i+\frac{w_i}{2}]$. We assume that $f(x)=0$ outside the interval $[0,1]$ and thus each agent will only choose a location from $\left[\frac{w_i}{2},1-\frac{w_i}{2} \right]$
		
	Let $\vec{x}$ denote the joint location profile $(x_1,x_2,...,x_n)$,
	and $\vec{x}_{-i}$  denote the profile without $i$, i.e. $(x_1,...,x_{i-1},x_{i+1}...,x_n)$.
	Given $\vec{x}$, let congestion function $c(x, \vec{x})$ be the the number of attraction intervals covering point $x$, 
	$$c(x,\vec{x})=\#\{x_i | x \in R_i\}$$
	
	Clearly, the following equation holds:
	\begin{gather*}
	c(x,\vec{x})=
	\begin{cases}
	c(x,\vec{x}_{-i})  & x\notin\left[x_i-\frac{w_i}{2},x_i+\frac{w_i}{2}\right] \\
	c(x,\vec{x}_{-i})+1 & x\in\left[x_i-\frac{w_i}{2},x_i+\frac{w_i}{2}\right]
	\end{cases}
	\end{gather*}
	
	For simplicity, we use $c(x)$ instead of $c(x,\vec{x})$ when there is no ambiguity.
	
	If a point $x$ is covered by multiple attraction intervals (i.e., $c(x)\geq 2$), then the support of that point is evenly divided among all these agents.
	Agent $i$'s support $s_i$ is then defined to be the total support from his attraction interval:
	\begin{gather*}
	s_i(\vec{x})=\int_{x_i-w_i/2}^{x_i+w_i/2}\frac{f(x)}{c(x)}\,\mathrm{d}x
	\end{gather*}
	
	In our model, we assume that the distribution function $f(x)$ and the width $w_i, \forall i$ are publicly known. We consider two kinds of utility settings: support utility and winner utility. The support utility setting uses the support function as agents' utility function. In the winner utility setting, only agents with the largest support are considered to be the winners, and share a total utility of 1 equally among them, while other agents have utility 0. Note that in the winner utility setting, each agent only cares about whether he is a winner and the number of winners, since if the agent is a winner, he has a higher utility when there are less winners.
	
	Formally, an attraction game is defined as follows:
	\begin{definition}
		Given the clients' distribution $f(x)$, an attraction game is a tuple $G=(N,w,L,u)$, where:
		\begin{itemize}
			\item $N=(1,2,\dots,n)$ is the set of all agents;
			\item $w=(w_1,w_2,\dots,w_n)$ is the widths associated to agents;
			\item $L=L_1\times L_2\times\dots\times L_n$ is the set of all possible location profiles, where $L_i=\left[\frac{w_i}{2},1-\frac{w_i}{2} \right]$.
			\item $u=(u_1,u_2,\dots,u_n)$ is the utility functions for the agents, the definition of which depends on the setting we consider:
			\begin{itemize}
				\item in the support utility setting, $u_i(\vec{x})=s_i(\vec{x})$;
				\item in the winner utility setting, 
				\begin{gather*}
				u_i(\vec{x})=
				\begin{cases}
				\frac{1}{|W|} & i\in W\\
				0 & \mathrm{otherwise}
				\end{cases}
				\end{gather*}
				where $W$ denotes the winner set.
			\end{itemize}
		\end{itemize}
	\end{definition}
	
	A Nash equilibrium of game $G$ is a stable location profile, where no agent can deviate to another location to get a higher utility.
	\begin{definition}[Nash equilibrium]
		Given a game $G$, the set of Nash equilibra $NE(G)$ contains all location profiles $\vec{x}$, such that $\forall i\in N$ and $\forall x'_i\in L_i$,
		\begin{gather*}
		u_i(x_i, \vec{x}_{-i})\geq u_i(x'_i, \vec{x}_{-i})
		\end{gather*}
	\end{definition}
	
	In different utility settings, the definitions of the utility functions are different, and thus have distinct Nash equilibria. Consider the following example:
	
	{\bf Example.} Assume that the clients are distributed uniformly.
	There are 3 agents with widths $w_1=0.4$, $w_2=0.3$, $w_1=0.4$. The location profile is $\vec{x}=(0.2, 0.65, 0.8)$ (see Figure~\ref{fig:example1}).

%	For congestion function, we have the following equation

%		\[
%c(x,\vec{x})=\left\{\begin{array}{ll}
%c(x,\vec{x}_{-i})  & x\notin[x_i-\frac{w_i}{2},x_i+\frac{w_i}{2}] \\
%c(x,\vec{x}_{-i})+1 & x\in[x_i-\frac{w_i}{2},x_i+\frac{w_i}{2}]\\
%\end{array}\right.
%\]

%		Let $R_i$ be the interval chosen by agent $i$ and $R=(R_1, R_2, \dots, R_n)$ be the action profile of all agents. Let $m(x)=\{i|x\in R_i\}$ and $c(x)=|m(x)|$. The utility function of agent $i$ is:
%	\begin{gather*}
%	u_i(R)=\int_{R_i}\frac{f(x)}{c(x)}\,\mathrm{d}x
%	\end{gather*}
	%Without loss of generality, we assume that the clients are distributed in some limited region $\Omega$, with probability density function $f(x)$. Let $N=(1,2,\dots,n)$ denote the set of agents. Each agent chooses a subset of $\Omega$ and obtain the probability measure of the chosen region as its utility. If a point is chosen by multiple agents, then the density of that point is evenly divided among all these agents.
%	%
%	\begin{itemize}
%	\item $x_i$ is the location.
%	\item $w_i$ is the width
%	\item $R_i$ is the interval
%	\item $u(R_i)$ is utility.
%	\item $c(x)$ is the number of available agent
%	\item $d$ is not necessary.
%	\end{itemize}
	\begin{figure}[h] %  figure placement: here, top, bottom, or page
	   \centering
	   \includegraphics[width=3.4in]{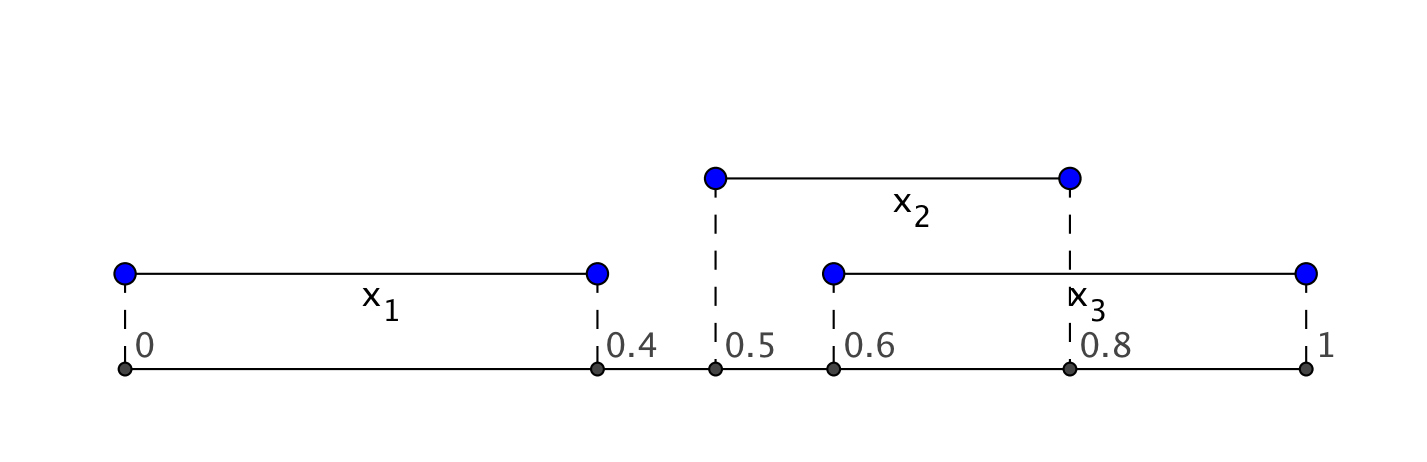} 
	   \caption{example}
	   \label{fig:example1}
	\end{figure}
 	The corresponding congestion function is
\begin{displaymath}
c(x)=\left\{\begin{array}{ll}
1  & x\in[0,0.4] \\
0 & x\in(0.4,0.5]\\
1 & x\in(0.5,0.6]\\
2 & x\in(0.6,0.8]\\
1 & x\in(0.8,1]\\
\end{array}\right.
\end{displaymath} 
	The support of the 3 agents are
	\begin{itemize}
	\item $s_1=\int^{0.4}_0 f(x)\,\mathrm{d}x=0.4$
	\item $s_2=\int^{0.6}_{0.5} f(x)\,\mathrm{d}x+\int^{0.8}_{0.6} \frac{f(x)}{2}\,\mathrm{d}x=0.2$
	\item $s_3=\int^{0.8}_{0.6} \frac{f(x)}{2}\,\mathrm{d}x+\int^{1}_{0.8} f(x)\,\mathrm{d}x=0.3$
	\end{itemize}
	
	In the support utility setting, the profile $\vec{x}$ does not form a Nash equilibrium, since given $\vec{x}_{-2}$, agent 2 has incentive to deviate to $0.55$. And by doing so, $u_2$ increases as $R_2$ has less intersection with $R_3$. However, in the winner utility setting, the profile $\vec{x}$ forms a Nash equilibrium and agent 1 is the unique winner.
	
%	\begin{itemize}
%	\item {\it support utility maximizer:}
%		Every agent wants to maximize his own support utility.
%		In Fig.~\ref{fig:example1}, given $\vec{x}_{-1}$, agent 1 has maximized the support utility.
%		While  from $0.65$.
%		By doing so, $R_2$ has less intersection with $R_3$, thus  .
%		
%	\item {\it winner utility maximization:}
%		The agents with the largest support utility are the winners.
%		The winners shares the winner utility 1 equally.
%		
%		In Fig.~\ref{fig:example1}, agent 1 is the unique winner with the largest support utility $0.4$.
%		Agent 1 has winner utility 1, and agent 2 and 3 has winner utility 0.
%		Both agent 2 and agent 3 cannot deviate alone to become a winner.
%		For example, if agent 3 deviates to $0.2$ same as agent 1, then agent 2 will become the winner.
%		
%	\end{itemize}
%	
%
%	If we consider the support utility setting, the example in Fig.~\ref{fig:example1} is not a Nash equilibrium.
%	If we consider the winner utility maximization setting, the example in Fig.~\ref{fig:example1} is a Nash equilibrium.	

	\section{Existence of Nash equilibrium in support utility setting}
	\label{sect:max}
	A well known theorem of Glicksberg\cite{glicksberg1952further} states that every continuous game has a mixed Nash equilibrium.
	By definition,
	$$u_i(x_i,\vec{x}_{-i})=\int_{x_i-w_i/2}^{x_i+w_i/2} \frac{f(x)}{c(x,\vec{x}_{-i})+1}\,\mathrm{d}x$$ 
	Agent $i$'s support utility is continuous with $x_i$.
	According to Glicksberg's theorem, there exists a mixed Nash equilibrium in this setting.
	
	However, due to the special structure of our model, we could further show that a pure Nash equilibrium always exists.
	The game can be viewed as a congestion game where the resources are the densities associated to each point. It is known that every finite congestion game has a pure strategy Nash equilibrium. Although agents's action space is infinitely in the game, we can still show that a pure strategy Nash equilibrium exists.
	\begin{theorem}
	There exists a pure Nash equilibrium in the support utility setting.
	\end{theorem}
	\begin{proof}
	Given other agents' locations $\vec{x}_{-i}$, agent i's support utility can be written as
	$$u_i(x_i)=\int^{x_i+\frac{w_i}{2}}_{x_i-\frac{w_i}{2}} \frac{f(x)}{c(x,\vec{x}_{-i})+1} \,\mathrm{d}x$$
	If agent $i$ prefers $x_i'$ to $x_i$, we have
	$$\int^{x_i'+\frac{w_i}{2}}_{x_i'-\frac{w_i}{2}} \,\frac{f(x)}{c(x,\vec{x}_{-i})+1} \,\mathrm{d}x>\int^{x_i+\frac{w_i}{2}}_{x_i-\frac{w_i}{2}} \frac{f(x)}{c(x,\vec{x}_{-i})+1} \,\mathrm{d}x$$
	On both sides of the inequality, we add the following term
	$$\int^{1}_{0} \sum_{k=1}^{c(x,\vec{x}_{-i})} \frac{f(x)}{k} \,\mathrm{d}x$$
	The left side of inequality becomes
	\begin{align*}
	&\int^{1}_{0}\, \sum_{k=1}^{c(x,\vec{x}_{-i})} \,\mathrm{d}x + \int^{x_i'+\frac{w_i}{2}}_{x_i'-\frac{w_i}{2}} \frac{f(x)}{c(x,\vec{x}_{-i})+1} \,\mathrm{d}x\\
	=&\int^{x_i-\frac{w_i}{2}}_{0} \,\sum_{k=1}^{c(x,\vec{x}_{-i})} \frac{f(x)}{k} \,\mathrm{d}x + \int^{1}_{x_i+\frac{w_i}{2}} \sum_{k=1}^{c(x,\vec{x}_{-i})} \frac{f(x)}{k}  \,\mathrm{d}x\\
	&+\int^{x_i'+\frac{w_i}{2}}_{x_i'-\frac{w_i}{2}} \sum_{k=1}^{c(x,(x_i',\vec{x}_{-i}))} \frac{f(x)}{k} \,\mathrm{d}x\\
	=&\int^{1}_{0} \,\sum_{k=1}^{c(x,(x_i',\vec{x}_{-i}))} \frac{f(x)}{k} \,\mathrm{d}x.
	\end{align*}
	Similarly, we can get the right side and the inequality becomes
	\begin{eqnarray}
	\int^{1}_{0} \sum_{k=1}^{c(x,(x_i',\vec{x}_{-i}))} \frac{f(x)}{k} \,\mathrm{d}x>	
	\int^{1}_{0} \sum_{k=1}^{c(x,\vec{x})} \frac{f(x)}{k} \mathrm{d}x
	\label{eqa:potential}
	\end{eqnarray}
	If we start from an arbitrary location profile, many agents' strategies are not optimal.
	Then in each round, we let one agent to better respond to other agents' strategies.
	By equation~(\ref{eqa:potential}), each time one agent improves his support utility, he actually improves a potential function,
	$$\Phi(\vec{x}) =\int^{1}_{0} \sum_{k=1}^{c(x,\vec{x})} \frac{f(x)}{k} \mathrm{d}x$$
%	
%	For any two actions $R_i$ and $R'_i$, define $R^{-}=R_i\setminus R'_i$ and $R^{+}=R'_i\setminus R_i$. We have
%	\begin{align*}
%	u_i(R'_i,R_{-i})-u_i(R_i,R_{-i})&=\int_{R^{+}}\frac{f(x)}{c(x)+1}\,\mathrm{d}x-\int_{R^{-}}\frac{f(x)}{c(x)}\,\mathrm{d}x \\
%	&=\Phi(R'_i,R_{-i})-\Phi(R_i,R_{-i})
%	\end{align*} 
%	Thus $\Phi(R)$ is indeed a potential function.
%	

	The potential function can be upper bounded,
%		\begin{eqnarray*}
%	\Phi(\vec{x}) &=& \int^{1}_{0} \sum_{k=1}^{c(x,\vec{x})} \frac{f(x)}{k} \mathrm{d}x\\
%	&\leq&\int^{1}_{0} \sum_{k=1}^{n} \frac{f(x)}{k} \mathrm{d}x\\
%	&=&\sum_{k=1}^{n}\frac{1}{k}
%	\end{eqnarray*}
$$\Phi(\vec{x}) = \int^{1}_{0} \sum_{k=1}^{c(x,\vec{x})} \frac{f(x)}{k} \mathrm{d}x \leq\int^{1}_{0} \sum_{k=1}^{n} \frac{f(x)}{k} \mathrm{d}x=\sum_{k=1}^{n}\frac{1}{k}.$$
	Combined the fact $\Phi(\vec{x})$ is continuous with $\vec{x}$, we have that $\Phi(\vec{x})$ has a maximum value.
	In the location profile $\vec{x}^*=\argmax_{\vec{x}} {\Phi(\vec{x}) }$, 
	no agent can improve his support utility	 and thus it is a Nash equilibrium.
		\end{proof}
%	\begin{align*}
%	\sum_{k=1}^{c(x)}d_x(k)=\sum_{k=1}^{c(x)}\frac{f(x)}{k}\le f(x)\sum_{k=1}^{n}\frac{1}{k}
%	\end{align*}
%	Then
%	\begin{align*}
%	\Phi(R)=\int_{\Omega}\sum_{k=1}^{c(x)}d_x(k)\,\mathrm{d}x\le \int_{\Omega}f(x)\left(\sum_{k=1}^{n}\frac{1}{k}\right)\,\mathrm{d}x=\sum_{k=1}^{n}\frac{1}{k}
%	\end{align*}
%	is limited. It follows that $\Phi(R)$ has a maximum value (the conditions need to be refined). Therefore
%	\begin{align*}
%	R^*=\argmax \Phi(R)
%	\end{align*}

	Noticing that the above proof does not use the fact of one-dimension space, and the result can be extended to multi-dimensional space.
	\begin{corollary}
	There is a pure Nash equilibrium in the support utility maximization setting if the location space is multi-dimensional.
	\end{corollary}
	
%	\subsection{How to compute the Nash equilibrium}
%It could take infinite number of steps to converge to the Nash equilibrium.
%??? new theorem
	
	\section{existence of nash equilibrium in winner utility setting}
	\label{sect:winner}
	In winner utility setting, the utility of agent $i$ is no longer continuous with respect to the agent's location.
	The potential function in the support-utility-maximizing setting does not work.
	
	\begin{definition} A game is an ordinal potential game,
 if there is a function 
$\phi :A \rightarrow R$
such that  $\forall a_{-i}\in A_{-i}$, $\forall {a'_{i}, a''_{i}\in A_{i}}$, 
$$(a'_{{i}},a_{{-i}})-u_{{i}}(a''_{{i}},a_{{-i}})>0\Leftrightarrow \Phi (a'_{{i}},a_{{-i}})-\Phi (a''_{{i}},a_{{-i}})>0$$
\end{definition}

It seems difficult to design a potential function to prove NE existence, the reason is the following theorem.
\begin{theorem}
Winner utility maximization game is not an ordinal potential game.
\end{theorem}
We prove by contradiction.
\begin{proof}
Consider the distribution function is
	\begin{gather*}
	f(x)=
	\begin{cases}
	4/3  & x\in[0,1/3)\\
	1/3  & x\in[1/3,2/3)\\
	4/3  & x\in[2/3,1]
	\end{cases}
	\end{gather*}
There are three agents with same width $w=1/3$.
We give two different paths from the location profile (1/6,1/6,1/6) to
(1/6,1/6,5/6). 

Path 1:
$(1/6, 1/6, 1/6)\rightarrow (1/6, 1/6, 5/9)\rightarrow
(1/6, 5/9, 5/9) \rightarrow
(1/6, 5/9, 5/6) \rightarrow
(1/6, 1/6, 5/6)
$

Path 2:
$(1/6, 1/6, 1/6)\rightarrow
(1/6, 1/6, 5/6)$.

Suppose there exists a potential function $\Phi$.
In Path 1, $u_3$ decreases in the first step, $u_2$ decreases in the second step. The support utility of the deviating agent does not change in the following steps.
By definition of $\Phi$, we should have $\Phi(1/6,1/6,5/6) < \Phi(1/6,1/6,1/6)$.

In Path 2, $u_3$ increases in the first step.
By definition, we should have $\Phi(1/6,1/6,5/6) > \Phi(1/6,1/6,1/6)$, contradiction.
\end{proof}

	In this setting, there exists a new strategy that an agent increase winner utility  by decreasing the support utility of both winner and himself. Consider the following example.
	
	{\bf Example.}
	Let the distribution $f(x)$ be
	\begin{displaymath}
	f(x)=\left\{\begin{array}{ll}
	5/4  & x\in[0,0.4] \\
	5/6 & x\in(0.4,1]\\
	\end{array}\right.
	\end{displaymath} 
	
		\begin{figure}[h] %  figure placement: here, top, bottom, or page
	   \centering
	   \includegraphics[width=3.4in]{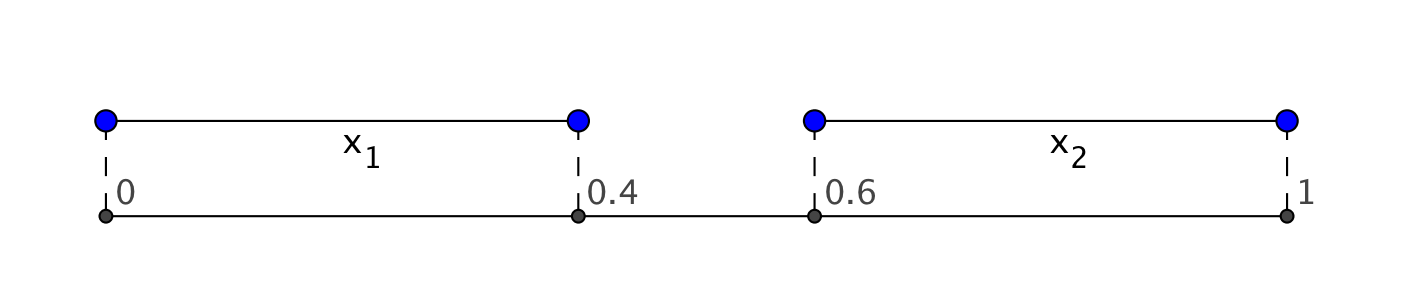} 
	   \caption{example}
	   \label{fig:2}
	\end{figure}
	There are two agents, and the location profile is $\vec{x}=(0.2,0.8)$. The width are equal $w_1=w_2=0.4$.
	In this case, agent 1 is the winner and $u_1=0.5$, $u_2=10/3$.
	However, agent 2 can move to $x_1$ and share the support from $[0,0.4]$ with agent 1.
	These two agents' support utility become $0.25$ and both agents are winners.
	Notice that agent 2 becomes a winner by decreasing both agents' support utility.

	When we consider winner utility maximization setting, we restrict to the case all agents have the same width $w_i=w$.
	First we prove a lemma which will be used frequently.
	This lemma roughly gives a situation where two agents have no incentive to deviate.
	\begin{lemma}
	Fix $k\geq 0$ agents' locations $\vec{x}$ at first\footnote{does not need to be a Nash equilibrium.}. Let $X$ be the set of maximizers of $$\int_{x-\frac{w}{2}}^{x+\frac{w}{2}}\frac{f(y)}{c(y,\vec{x})+1}\,\mathrm{d}y$$
	%A new agent has a set of best choices to maximize his support utility. 
	Suppose there are two new agents A and B. If both two agents choose the two locations $x_A,x_B\in X$($x_A$ and $x_B$ could be the same) simultaneously, then both agents have the same support utility,
	and both A and B cannot have more support utility than the other by changing location.
	\label{lemma:wzh1}
	\end{lemma}
	\begin{proof}
	By definition of $x_A$ and $x_B$, we have $$\int_{x_A-\frac{w}{2}}^{x_A+\frac{w}{2}}\frac{f(x)}{c(x)+1}\,\mathrm{d}x=\int_{x_B-\frac{w}{2}}^{x_B+\frac{w}{2}}\frac{f(x)}{c(x)+1}\,\mathrm{d}x$$
	When A and B are located at the same time, their attraction interval may overlap.
	This will decrease the support from clients in the intersection interval, 
	but the decrements in two support utility are the same.
	We use $R_A$ denote the attraction interval $[x_A-\frac{w}{2},x_A+\frac{w}{2}]$ and 
	$R_B$ denote the interval $[x_B-\frac{w}{2},x_B+\frac{w}{2}]$.
	Formally, agent A's support utility will be
	\begin{align*}
	&\int_{R_A-R_B}\frac{f(x)}{c(x)+1}\,\mathrm{d}x+\int_{R_A\cap R_B}\frac{f(x)}{c(x)+2}\,\mathrm{d}x\\
	=&\int_{R_A}\frac{f(x)}{c(x)+1}\,\mathrm{d}x-\int_{R_A\cap R_B}f(x)(\frac{1}{c(x)+2}-\frac1{c(x)+1})\,\mathrm{d}x\\
	=&\int_{R_B}\frac{f(x)}{c(x)+1}\,\mathrm{d}x-\int_{R_A\cap R_B}f(x)(\frac{1}{c(x)+2}-\frac1{c(x)+1})\,\mathrm{d}x\\	
	=&\int_{R_B-R_A}\frac{f(x)}{c(x)+1}\,\mathrm{d}x+\int_{R_A\cap R_B}\frac{f(x)}{c(x)+2}\,\mathrm{d}x.
	\end{align*}
	which is same as agent B's support utility. 
	
	For the second part, we prove by contradiction.
	Suppose A moves to $x_{A'}$ and gets more support utility than B, then
%	\begin{eqnarray*}
%	&\int_{R_{A'}-R_B}\frac{f(x)}{c(x)+1}\,\mathrm{d}x+\int_{R_{A'}\cap R_B}\frac{f(x)}{c(x)+2}\,\mathrm{d}x&>\\
%	&\int_{R_B-R_{A'}}\frac{f(x)}{c(x)+1}\,\mathrm{d}x+\int_{R_{A'}\cap R_B}\frac{f(x)}{c(x)+2}\,\mathrm{d}x&\\
%	\Rightarrow 	&\int_{R_{A'}-R_B}\frac{f(x)}{c(x)+1}\,\mathrm{d}x+\int_{R_{A'}\cap R_B}\frac{f(x)}{c(x)+1}\,\mathrm{d}x&>\\
%	&\int_{R_B-R_{A'}}\frac{f(x)}{c(x)+1}\,\mathrm{d}x+\int_{R_{A'}\cap R_B}\frac{f(x)}{c(x)+1}\,\mathrm{d}x&\\
%	\Rightarrow	
%	&\int_{R_{A'}}\frac{f(x)}{c(x)+1}\,\mathrm{d}x>\int_{R_B}\frac{f(x)}{c(x)+1}\,\mathrm{d}x&
%	\end{eqnarray*}

\begin{align*}
	&\int_{R_{A'}-R_B}\frac{f(x)}{c(x)+1}\,\mathrm{d}x+\int_{R_{A'}\cap R_B}\frac{f(x)}{c(x)+2}\,\mathrm{d}x>\\
	&\int_{R_B-R_{A'}}\frac{f(x)}{c(x)+1}\,\mathrm{d}x+\int_{R_{A'}\cap R_B}\frac{f(x)}{c(x)+2}\,\mathrm{d}x\\
	\Rightarrow 	&\int_{R_{A'}-R_B}\frac{f(x)}{c(x)+1}\,\mathrm{d}x+\int_{R_{A'}\cap R_B}\frac{f(x)}{c(x)+1}\,\mathrm{d}x>\\
	&\int_{R_B-R_{A'}}\frac{f(x)}{c(x)+1}\,\mathrm{d}x+\int_{R_{A'}\cap R_B}\frac{f(x)}{c(x)+1}\,\mathrm{d}x\\
	\Rightarrow	
	&\int_{R_{A'}}\frac{f(x)}{c(x)+1}\,\mathrm{d}x>\int_{R_B}\frac{f(x)}{c(x)+1}\,\mathrm{d}x
\end{align*}
	which contradicts to the fact that $x_B$ is a best location. 
 Since A and B are symmetric, agent B cannot get more support utility than A neither. 
  	\end{proof}
	
	When there are 2 agents, we put both agents at the same position where maximizes $\int_{[x-w/2,x+w/2]}f(y)\,\mathrm{d}y$. By setting $k=0$ in Lemma~\ref{lemma:wzh1}, we know both agents are winners.
	Either one cannot become the unique winner by deviating.
	Thus the location profile constitutes a Nash equilibrium. 
	\begin{theorem}
	There is a pure Nash equilibrium when there are 2 agents.
	\end{theorem}

	When there are 3 agents, the problem becomes quite complicated. 
	Since there is no symmetric property, there are many cases to consider when we are checking the stable equilibrium.
	We propose Algorithm 1 for 3 agents.
	\begin{itemize}
	\item Let $u_1$ denote the largest support utility that the first player could achieve, i.e.  $u_1$ is the maximum value of $\int_{x-w/2}^{x+w/2}f(y)\,\mathrm{d}y$.
	\item {\it Case 1:} If we can allocate three agents at the same time such that everyone achieves the support utility $u_1$, then we allocate them at those three locations.

	\item {\it Case 2:} If we can allocate only two agents at the same time such that every one achieves the support utility $u_1$, then we allocate agent 1 at one of the two locations, the other two agents together at the other location.
	\item Otherwise, we can allocate only one agent that achieves the support utility $u_1$. There exists a set of locations that agent 1 achieves the support utility $u_1$, we allocate agent 1 at the leftmost one, denoted by $x_1$.
			Given agent 1's location, let the largest support utility for agent 2 is $u_2$.
			\begin{itemize}
				\item {\it Case 3:} If agent 2 can achieve support utility $u_2$ at location $x_1$, we allocate agents 2 and 3 together at $x_1$.
				\item {\it Case 4:} Otherwise, if on both left side and right side of agent 1, there exists locations where agent 2 achieves support utility $u_2$. We allocate agent 2 at the rightmost position in the left part and agent 3 at the leftmost position in the right part, i.e., 
				$$x_2=\max_{x<x_1}\left\{ x\,\bigg| \int_{x-\frac{w}{2}}^{x+\frac{w}{2}} \frac{f(y)}{c(y)+1}\,\mathrm{d}y=u_2\right\},$$ $$x_3=\min_{x>x_1}\left\{ x\,\bigg| \int_{x-\frac{w}{2}}^{x+\frac{w}{2}} \frac{f(y)}{c(y)+1}\,\mathrm{d}y=u_2\right\}.$$
				\item {\it Case 5:} Otherwise, the locations that maximize agent 2's support utility lie on the same side of agent 1. We allocate agent 2 to the closest position and agent 3 at the farthest position. Let
				$$t_2=\min\left\{ x\,\bigg| \int_{x-\frac{w}{2}}^{x+\frac{w}{2}} \frac{f(y)}{c(y)+1}\,\mathrm{d}y=u_2\right\},$$ $$t_3=\max\left\{ x\,\bigg| \int_{x-\frac{w}{2}}^{x+\frac{w}{2}} \frac{f(y)}{c(y)+1}\,\mathrm{d}y=u_2\right\}.$$
				If $t_2<x_1$, then set $x_2=t_3$, $x_3=t_2$. Otherwise, set $x_2=t_2$, $x_3=t_3$.
			\end{itemize}
	\end{itemize}	
	\begin{theorem}
	When there are 3 agents, Algorithm 1 gives a pure Nash equilibrium.
	\end{theorem}	
	
	Here is the idea of Algorithm 1.
	In most of time we allocate agent 1 where he gets the largest support utility. 
	Then we allocate agent 2 and 3 to get the largest support utility as possible.
	Agent 2 and 3 have the same support utility and hinder each other.
	If agent 2 wants to get the same support utility as agent 1 by decreasing both $u_1$ and $u_2$, then $u_3$ becomes the largest, and agent 3 is the unique winner.
	Thus the location profile forms a Nash equilibrium.

	\begin{proof}
	The proof follows algorithm's structure. 
	In each case, we consider who wins and whether the agents' attraction intervals intersect. In all cases, we prove no one has incentive to deviate.

	%We consider the first case that we can allocate 3 agents such that everyone achieves the maximum support utility $u_1$.
		In Case 1, winner set is $\{1,2,3\}$. 
		Keeping agent 2's location fixed, by Lemma~\ref{lemma:wzh1}, agent 1 cannot get more support utility than agent 3.
		%Keeping agent 3's location fixed, by Lemma~\ref{lemma:wzh1}, agent 1 cannot get more support utility than agent 2.
		
		Thus if agent 1 moves, the other two agents have at least the same support utility.
		Using similar arguments, no one gets more support utility than any other player.
		So on one	has incentive to deviate.
		
		%Then we consider the second case.
 %let $x_1$ be the first agent's position and $x_2$ be the second agent's position.
		In Case 2, winner set is $\{1\}$.
		Since agent 2 and 3 are at the same location, we only need to prove agent 2 does not have incentive to deviate.
%		By Lemma~\ref{lemma:wzh1}, after the second agent moves, the third agent have at least the same support utility.		
		We prove by contradiction. Let agent 2 could become a winner by deviating to $x_2'$. Let $R_2'$ denote the corresponding attraction interval.
		If $R_2'$ does not intersect with $R_1$, agent 1 always has more support utility than agent 2. Then $R_2'$ intersects with $R_1$. For same reason, $R_2'$ intersects with $R_3$.
		Without agent 3, agent 2 has at most the same support utility as agent 1.
		But $R_3$ only has intersection with $R_2'$, this intersection makes agent 2 have strictly less support utility than the agent 1. 
		So agent 2 cannot become a winner and nobody has incentive to deviate.

		In Case 3, winner set is $\{1,2,3\}$. 
		Since they have the same location, we only need to prove agent 3 has no incentive to deviate.
		Suppose agent 3 has incentive to deviate to $x_3'$ with attraction interval $R_3'$, then he must become the unique winner.
		Formally,
%		\begin{eqnarray*}
%		&\int_{R_3'\cap R_1} \frac{f(x)}{3}\,\mathrm{d}x+\int_{R_3'-R_1}f(x)\,\mathrm{d}x \geq&\\
%		&\int_{R_3'\cap R_1} \frac{f(x)}{3}\,\mathrm{d}x+\int_{R_1-R_3}\frac{f(x)}{2}\,\mathrm{d}x&\\
%		\Rightarrow &\int_{R_3'\cap R_1} \frac{f(x)}{2}\,\mathrm{d}x+\int_{R_3'-R_1}f(x)\,\mathrm{d}x \geq&\\
%		&\int_{R_3'\cap R_1} \frac{f(x)}{2}\,\mathrm{d}x+\int_{R_1-R_3}\frac{f(x)}{2}\,\mathrm{d}x&
%		\end{eqnarray*}
	
	\begin{align*}
		&\int_{R_3'\cap R_1} \frac{f(x)}{3}\,\mathrm{d}x+\int_{R_3'-R_1}f(x)\,\mathrm{d}x \geq\int_{R_3'\cap R_1} \frac{f(x)}{3}\,\mathrm{d}x+\int_{R_1-R_3}\frac{f(x)}{2}\,\mathrm{d}x&\\
		\Rightarrow &\int_{R_3'\cap R_1} \frac{f(x)}{2}\,\mathrm{d}x+\int_{R_3'-R_1}f(x)\,\mathrm{d}x \geq\int_{R_3'\cap R_1} \frac{f(x)}{2}\,\mathrm{d}x+\int_{R_1-R_3}\frac{f(x)}{2}\,\mathrm{d}x&
	\end{align*}
		Consider the situation when only agent 1 has been located.
		The left side of the inequality is agent 2's support utility by choosing $x_3'$.
		The right side of the inequality is agent 2's support utility by choosing $x_1$.
		That means agent 2 gets more support utility by choosing $x_3'$ than choosing $x_1$, contradicting to the assumption.
		Hence agent 3 has no incentive to deviate.

		In Case 4, we have $u_2=u_3$ by Lemma~\ref{lemma:wzh1}. 
		There are 3 possibilities about who the winners are.
		
		{\it Case 4.1:}Winner set is $\{2,3\}$. 
		Agent 2 and 3 have no incentive to deviate by Lemma~\ref{lemma:wzh1}.
		We next prove agent 1 has no incentive to deviate.
		If $R_2\cap R_1=\emptyset$, agent 3 cannot have more support utility than agent 1.
		So we have $R_2\cap R_1\neq \emptyset$ and $R_1\cap R_3\neq \emptyset$.
		We consider whether $R_2\cap R_3$ is empty.
		
		{\it Case 4.1.1:} $R_2\cap R_3\neq \emptyset$.
	 	\begin{figure}[h] %  figure placement: here, top, bottom, or page
	   \centering
	   \includegraphics[width=3in]{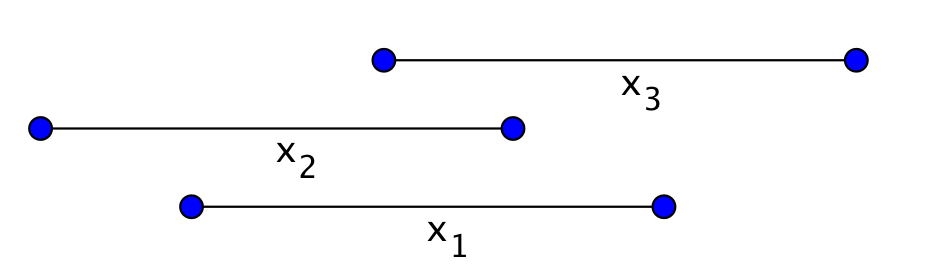} 
	   \caption{example}
	   \label{fig:example3}
	\end{figure}
		Suppose agent 1 benefits by deviating to $x_1'$. Let $R_1'$ denote the new attraction interval.
		If $R_1'\cap R_3=\emptyset$, then we have
		\begin{align*}
		u_1'&=\int_{R_1'-R_2}f(x)\,\mathrm{d}x+\int_{R_1'\cap R_2}\frac{f(x)}{2}\,\mathrm{d}x\\
		&\leq \int_{R_1'-R_1}f(x)\,\mathrm{d}x+\int_{R_1'\cap R_1}\frac{f(x)}{2}\,\mathrm{d}x\\
		&\leq\int_{R_3\cap R_1}\frac{f(x)}{2}\,\mathrm{d}x+\int_{R_3- R_1}f(x)\,\mathrm{d}x\\
		&<\int_{R_3\cap R_2}\frac{f(x)}{2}\,\mathrm{d}x+\int_{R_3- R_2}f(x)\,\mathrm{d}x
		\end{align*}
		This is the agent 3's support utility after agent 1's deviation. Agent 1 cannot become a winner.
		
		If $R_1'\cap R_3\neq \emptyset$ and $R_1'\cap(R_3-R_2)=\emptyset$, then we have
		\begin{align*}
		u_1'&=\int_{R_1'-R_2}f(x)\,\mathrm{d}x+\int_{R_2-R_3}\frac{f(x)}{2}\,\mathrm{d}x+\int_{R_1'\cap R_3}\frac{f(x)}{3}\,\mathrm{d}x\\
		&\leq\int_{R_2-R_3}\frac{f(x)}{2}\,\mathrm{d}x+\int_{R_1'-R_3}\frac{f(x)}{3}\,\mathrm{d}x+\int_{R_2-R_1'}\frac{f(x)}{2}\,\mathrm{d}x\\
		&<\int_{R_2-R_1}f(x)\,\mathrm{d}x+\int_{R_1-R_3}\frac{f(x)}{2}\,\mathrm{d}x+\int_{R_1'-R_3}\frac{f(x)}{3}\,\mathrm{d}x+\int_{R_2-R_1'}\frac{f(x)}{2}\,\mathrm{d}x\\
		&=\int_{R_3-R_1}f(x)\,\mathrm{d}x+\int_{R_1-R_2}\frac{f(x)}{2}\,\mathrm{d}x+\int_{R_1'-R_3}\frac{f(x)}{3}\,\mathrm{d}x+\int_{R_2-R_1'}\frac{f(x)}{2}\,\mathrm{d}x
		\end{align*}
		This is the agent 3's support utility after agent 1's deviation.
		
		If $R_1'\cap(R_3-R_2)\neq\emptyset$. Since agent 1 has at least the same support utility as agent 2, we have
		\begin{align*}
%		&&u_1'=\int_{R_1'-R_2}\frac{f(x)}{2}\,\mathrm{d}x+\int_{R_1'\cap R_2\cap R_3}\frac{f(x)}{3}\,\mathrm{d}x+\int_{R_1'-(R_2\cap R_3}\frac{f(x)}{2}\,\mathrm{d}x\\
%		&&u_1'=\int_{R_1'-R_2}\frac{f(x)}{2}\,\mathrm{d}x+\int_{R_1'\cap R_2\cap R_3}\frac{f(x)}{3}\,\mathrm{d}x+\int_{R_1'-(R_2\cap R_3}\frac{f(x)}{2}\,\mathrm{d}x\\
		\int_{R_1'-R_2}\frac{f(x)}{2}\,\mathrm{d}x\geq \int_{R_2-R_1'}f(x)\,\mathrm{d}x,
		\end{align*}
		contradicting to the definition of agent 2's location.
		To sum up, agent 1 cannot become a winner by deviating in Case 4.1.1.
		
		{\it Case 4.1.2:} $R_2\cap R_3= \emptyset$. The proof is similar to the that in Case 4.1.1 and thus omitted.
		
		{\it Case 4.2:}Winner set is $\{1\}$.
		Then agent 1 has no incentive to move.
		By definition of agent 2's location, we have
%		\begin{eqnarray*}
%		&\int_{R_2-R_1}f(x)\,\mathrm{d}x+\int_{R_1\cap R_2}\frac{f(x)}{2}\,\mathrm{d}x> \int_{R_1}\frac{f(x)}{2}\,\mathrm{d}x&\\
%		&\int_{R_2-R_1}f(x)\,\mathrm{d}x> \int_{R_1-R_2}\frac{f(x)}{2}\,\mathrm{d}x&
%		\end{eqnarray*}
	\begin{gather*}
		\int_{R_2-R_1}f(x)\,\mathrm{d}x+\int_{R_1\cap R_2}\frac{f(x)}{2}\,\mathrm{d}x> \int_{R_1}\frac{f(x)}{2}\,\mathrm{d}x\\
		\int_{R_2-R_1}f(x)\,\mathrm{d}x> \int_{R_1-R_2}\frac{f(x)}{2}\,\mathrm{d}x
	\end{gather*}
%	\begin{figure}[h] %  figure placement: here, top, bottom, or page
%	   \centering
%	   \includegraphics[width=3in]{004_1.png} 
%	   \caption{example}
%	   \label{fig:example1}
%	\end{figure}
		We claim that $R_1\cap R_2 \cap R_3=\emptyset$, otherwise we have
%		\begin{eqnarray*}
%		u_2&=&\int_{R_2-R_1}f(x)\,\mathrm{d}x+\int_{R_2-(R_1\cap R_3)}\frac{f(x)}{2}\,\mathrm{d}x+\int_{R_2\cap R_1\cap R_3}\frac{f(x)}{3}\,\mathrm{d}x\\
%		&>&\int_{R_1-R_2}\frac{f(x)}{2}\,\mathrm{d}x+\int_{R_2-(R_1\cap R_3)}\frac{f(x)}{2}\,\mathrm{d}x+\int_{R_2\cap R_1\cap R_3}\frac{f(x)}{3}\,\mathrm{d}x\\
%		&=&u_1
%		\end{eqnarray*}
		\begin{align*}
		u_2&=\int_{R_2-R_1}f(x)\,\mathrm{d}x+\int_{R_2-(R_1\cap R_3)}\frac{f(x)}{2}\,\mathrm{d}x+\int_{R_2\cap R_1\cap R_3}\frac{f(x)}{3}\,\mathrm{d}x\\
		&>\int_{R_1-R_2}\frac{f(x)}{2}\,\mathrm{d}x+\int_{R_2-(R_1\cap R_3)}\frac{f(x)}{2}\,\mathrm{d}x+\int_{R_2\cap R_1\cap R_3}\frac{f(x)}{3}\,\mathrm{d}x\\
		&=u_1
		\end{align*}
		Suppose agent 2 deviates to $x_2'$.
		If $x_2'<x_2$, $u_2$ weakly decreases and $R_1\cap R_2$ weakly shrinks. Furthermore, $u_1$ weakly increases. Agent 2 would not be the winner.
		If $x_2'\in (x_2,x_3)$, by the definition of $x_2$ and $x_3$, agent 2 has less support utility than agent 3 no matter whether $R_2\cap R_3$ is empty.
		If $x_2'\in [x_3,1)$, $R_2'\cap R_3$. By Lemma~\ref{lemma:wzh1}, $u_2'$ is at most equal to agent 3's support utility.
		Now we prove agent 1 has strictly higher support utility than agent 3 after agent 2's deviation.
		Agent 1's support utility is
		\begin{align*}
		&\int_{R_1-R_3}f(x)\,\mathrm{d}x+\int_{R_3-R_2'}f(x)\,\mathrm{d}x+\int_{R_1\cap R_3\cap R_2'}\frac{f(x)}{3}\,\mathrm{d}x\\
		>&\int_{R_3-R_1}f(x)\,\mathrm{d}x+\int_{R_3-R_2'}\frac{f(x)}{2}\,\mathrm{d}x+\int_{R_1\cap R_3\cap R_2'}\frac{f(x)}{3}\,\mathrm{d}x\\
		>&\int_{R_3-R_1}\frac{f(x)}{2}\,\mathrm{d}x+\int_{R_3-R_2'}\frac{f(x)}{2}\,\mathrm{d}x+\int_{R_1\cap R_3\cap R_2'}\frac{f(x)}{3}\,\mathrm{d}x
		\end{align*}
		This is the agent 3's support utility.
		To sum up, agent 2 would not deviate, neither does agent 3.
		
		{\it Case 4.3:} winner set is $\{1,2,3\}$.
		The proof of agent 1 would not deviate is similar to that in Case 4.1.
		The proof of agent 2 or 3 would not deviate is similar to that in Case 4.2.

		{\it Case 5:} first we claim winner set can not be $\{2,3\}$.
		By definition, we have $R_3\cap R_1 \subset R_3\cap R_2$.
		We consider $u_1,u_2$ when there are only agent 1 and 2, then we take count into the impact of agent 3.
		\begin{align}
		u_1=&\int_{R_1-R_2}f(x)\,\mathrm{d}x+\int_{R_1\cap R_2}\frac{f(x)}{2}\,\mathrm{d}x\nonumber\\
		&+\int_{R_1\cap R_2\cap R_3}\left(\frac{f(x)}{3}-\frac{f(x)}{2}\right)\,\mathrm{d}x\label{eq1}\\
		u_2=&\int_{R_2-R_1}f(x)\,\mathrm{d}x+\int_{R_1\cap R_2}\frac{f(x)}{2}\,\mathrm{d}x\nonumber\\
		&+\int_{R_1\cap R_2\cap R_3}\left(\frac{f(x)}{3}-\frac{f(x)}{2}\right)\,\mathrm{d}x\nonumber\\
		&+\int_{(R_2\cap R_3)-R_1}\left(\frac{f(x)}{2}-f(x)\right)\,\mathrm{d}x\label{eq2}
		\end{align}
		We have $u_1\geq u_2$. Thus if agent 2 is a winner, agent 1 is a winner too.
		
		{\it Case 5.1:} winner set is $\{1\}$.
		In this sub-case, the argument is independent with the leftmost property of $x_1$. Then w.l.o.g. we assume $x_1<x_2\leq x_3$.
		Agent 1 has no incentive to deviate. Suppose agent 2 deviate to $x_2'$.
		If $x_2'\in[0,x_2)\cup (x_3,1] $, agent 3 has strictly more support utility than agent 2.
		If $x_2'\in(x_2,x_3]$, agent 1's support utility weakly increase, agent 2's support utility is weakly smaller than agent 3's.
		Agent 3's support utility weakly decreases. Then agent 1 has strictly more support utility than agent 2.
		Thus agent 2 would not deviate.
		Suppose agent 3 deviates to $x_3'$.
		If $x_3'\in [0,x_2)\cup (x_3,1]$, then agent 2 has strictly more support utility than agent 3.
		If $x_3'\in[x_2,x_3)$, agent 3 has weakly less support utility than agent 2.
		But the support utility difference between agent 2 and agent 1 becomes larger.
		Thus agent 3 would not deviate.

		{\it Case 5.2:} winner set is $\{1,2,3\}$.
		The proof that agent 2 and 3 have no incentive to deviate is same as in Case 5.1.
		Since $u_1=u_2$ and by Equation (\ref{eq1}) and (\ref{eq2}), we have $R_2\cap R_3=\emptyset$.
		Since $u_1=u_2$, $x_2$ is also a best choice for agent 1 at the beginning. 
		By the leftmost property of $x_1$, we know $x_1<x_2$.
		Suppose agent 1 deviates to $x_1'$.
		If $x_1'\in[0,x_1)\cup (x_1,x_2)$, agent 1 has strictly less support utility than agent 2.
		If $x_1'\in[x_2,1)$ and $R_1'\cap R_3=\emptyset$, agent 1 has strictly less support utility than agent 3.
		\begin{align*}
		u_1'=&\int_{R_1'-R_2}f(x)\,\mathrm{d}x+\int_{R_1'\cap R_2}\frac{f(x)}{2}\,\mathrm{d}x\\
		<&\int_{R_1'-R_1}f(x)\,\mathrm{d}x+\int_{R_1'\cap R_1}\frac{f(x)}{2}\,\mathrm{d}x \\
		\leq&u_3
		\end{align*}
		If $x_1'\in[x_2,1)$ and $R_1'\cap R_3\neq \emptyset$, 
		\begin{align*}
		u_1'&=\int_{R_1'-R_2-R_3}f(x)\,\mathrm{d}x+\int_{R_1'\cap (R_2\cup R_3)}\frac{f(x)}{2}\,\mathrm{d}x\\
		&<\int_{R_1'-R_2}f(x)\,\mathrm{d}x+\int_{R_1'\cap R_2}\frac{f(x)}{2}\,\mathrm{d}x
		\end{align*}
		This is agent 2's support utility after deviation.
		Agent 1 has strictly less support utility than agent 2.
		Hence in case 5.2, no agent would deviate.
%		There are three sub-cases to consider. 
%		(1)When $x_2=x_1$, all three agents are at the same location. Everyone is the winner.
%		If this is not a Nash equilibrium, then the first agent has incentive to move.
%		Suppose the second agent will move to $x_2'$ and he is the only winner now. We should have:
%\begin{eqnarray*}
%&&\int_{R_{2'}\cap R_1}\frac{f(x)}3\,\mathrm{d}x+\int_{R_{2'}-R_1}f(x)\,\mathrm{d}x>\int_{R_{2'}\cap R_1}\frac{f(x)}3\,\mathrm{d}x+\int_{R_1-R_{2'}}\frac{f(x)}2\,\mathrm{d}x\\
%&\Rightarrow&\int_{R_{2'}\cap R_1}\frac{f(x)}2\,\mathrm{d}x+\int_{R_{2'}-R_1}f(x)\,\mathrm{d}x>\int_{R_{2'}\cap R_1}\frac{f(x)}2\,\mathrm{d}x+\int_{R_1-R_{2'}}\frac{f(x)}2\,\mathrm{d}x\\
%&\Rightarrow&\int_{R_{2'}\cap R_1}\frac{f(x)}2\,\mathrm{d}x+\int_{R_{2'}-R_1}f(x)\,\mathrm{d}x>\int_{R_1}\frac{f(x)}2\,\mathrm{d}x
%\end{eqnarray*}
%This means $x_2'$ gives strictly more support utility than $x_2$, contradiction.
%(2)
	\end{proof}

When $w=0.5$, the attraction intervals overlap in general.
We make use of this property and give Algorithm 2 to construct a Nash equilibrium with width $0.5$.

\begin{itemize}
	\item Let $u_1$ denote the largest support utility that agent 1 could achieve, i.e.  $u_1=\max\{ \int_{[x-w/2,x+w/2]} f(y) \,\mathrm{d}y\}$.
	\item If we can allocate 2 agents at the same time such that everyone achieves the support utility $u_1$, i.e. 
$$\int_{[0,0.5]}f(x)\,\mathrm{d}x=\int_{[0,0.5]}f(x)\,\mathrm{d}x=u_1.$$ When there are $k$ agents, we allocate $\lceil k/2 \rceil$ agents at $0.25$ and 
$k- \lceil k/2 \rceil$ agents at $0.75$.
	\item If we can allocate only one agent that achieves the support utility $u_1$, then we allocate the first agent at $x_1$. Define $t$ such that everyone in the first t agents maximizes the support utility at $x_1$ given the previous agents' locations, but this no longer holds for the $(t+1)$-th agent. 
	\begin{itemize}
	\item When there are $k\leq t+1$ agent, we allocate them together at $x_1$.
	\item When there are $k\geq t+2$ agent, let $\vec{x}=(x,...,x)$ (the number of $x$ is $k-2$). 
	We define left largest support utility $ll(x)$ and right largest support utility $rl(x)$:
	$$ll(x)=\max\left\{\int^{z+0.25}_{z-0.25}\frac{f(y)}{c(y,\vec{x})}\,\mathrm{d}y \,\bigg|\,z\leq x\right\}$$
	$$rl(x)=\max\left\{\int^{z+0.25}_{z-0.25}\frac{f(y)}{c(y,\vec{x})}\,\mathrm{d}y \,\bigg|\,z\geq x\right\}$$	
	Here $c(y,\vec{x})=k-2$ if $|y-x|\leq w/2$ and  $c(y,\vec{x})=0$ if $|y-x|>w/2$.
	There exists $x^*$ such that $ll(x^*)=rl(x^*)$. 
	Let $x_l\leq x$ be a solution  of $z$ for $\int^{z+0.25}_{z-0.25}\frac{f(y)}{c(y)}\,\mathrm{d}y=ll(x^*)$
	and $x_r\geq x$ be a solution  of $z$ for $\int^{z+0.25}_{z-0.25}\frac{f(y)}{c(y)}\,\mathrm{d}y=rl(x^*)$.
	We put the first $k-2$ agents at $x^*$, $(k-1)$-th agent at $x_l$, $k$-th agent at $x_r$.
	\end{itemize}
\end{itemize}

	In the second case and when there are more than $k\geq t+2$ agents, based on the width is $0.5$ the union of the support of $(k-1)$-th agent and $k$-th agent will cover the support of previous agents.
	 In fact, the $(k-1)$-th and $k$-th agents are the unique two winners. 

	\begin{theorem}
	When $w=0.5$, Algorithm 2 gives a Nash equilibrium.
	\end{theorem}
	The main proof is omitted due to the space. We only prove that Algorithm 2 could find a location profile, i.e., there exists $x^*$ such that $ll(x^*)=rl(x^*)$ when there are $k\geq t+2$ agents in the second case.	
	
	\begin{proof}
	
	When the first $k-2$ agents located at $x_1$, suppose the support utility of the $(k-1)$th agent is maximized at $x_2$.
	In fact, we can prove $x_2\neq x_1$. W.l.o.g, we assume $x_2>x_1$, then $rl(x_1)\geq ll(x_1)$.
	Moreover $$rl(0.75)=\int^{1}_{0.5}\frac{f(y)}{k-1}\,\mathrm{d}y\leq \int^{x_1+0.25}_{x_1-0.25}\frac{f(y)}{k-1}\,\mathrm{d}y\leq ll(0.75)$$
	Since $ll(x)$ is continuous and weakly increasing while $rl(x)$ is continuous and weakly decreasing, 
	$x^*$ exists.
	\end{proof}
	For smaller width, the attraction interval may not overlap.
	This results new possibilities of the interval intersection relationship,
	and many possibilities about who is the winner. 
	Hence the proof does not hold for smaller width.

	\begin{lemma}
	If there always exists a Nash equilibrium when $w=0.5$, then it also holds for $w\geq 0.5$.
	\end{lemma}
	In this case, the interval $(1-w,w)$ belongs to every agent's support.
	The idea is we can remove this interval and the problem becomes proving the existence of Nash equilibrium with $w=0.5$.
	\begin{proof}
	When $w>0.5$, we can construct a Nash equilibrium from an instance with $w=0.5$. We let the new distribution function be	
	\begin{gather*}
g(x)=\begin{cases}
\displaystyle \frac{f((2-2w)x)}{\int^{1-w}_0f(y)\,\mathrm{d}y+\int^1_wf(y)\,\mathrm{d}y} & \displaystyle x\leq 0.5 \\
\qquad\\
\displaystyle\frac{f((2-2w)x+2w-1)}{\int^{1-w}_0f(y)\,\mathrm{d}y+\int^1_wf(y)\,\mathrm{d}y} & \displaystyle x> 0.5
\end{cases}
\end{gather*}
%	\begin{displaymath}
%g(x)=\left\{\begin{array}{ll}
%\frac{f((2-2w)x)}{\int^{1-w}_0f(y)\,\mathrm{d}y+\int^1_wf(y)\,\mathrm{d}y} &  x\leq 0.5 \\
%\frac{f((2-2w)x+2w-1)}{\int^{1-w}_0f(y)\,\mathrm{d}y+\int^1_wf(y)\,\mathrm{d}y} &  x> 0.5
%\end{array}\right.
%\end{displaymath}
	Suppose there is a Nash equilibrium $(x_1,x_2,...,x_n)$ under distribution $g$ with width $0.5$. We can verify that
	$((2-2w)x_1+w-0.5, (2-2w)x_2+w-0.5,...(2-2w)x_n+w-0.5)$ forms a Nash equilibrium in the distribution $f$ and with width
	$w$.
	\end{proof}
	
\section{Fairness in the Nash equilibrium}
\label{sect:fair}
\begin{definition}[Fairness]
	Given a game $G$, define the fairness of the game to be:
%	\begin{gather*}
	$$FAIR=\min_{\vec{x}\in NE(G)}\frac{\min_{i}u_i(\vec{x})}{\max_{i}u_i(\vec{x})}$$
	%\end{gather*}
\end{definition}

Intuitively, given a location profile $\vec{x}$, the ratio $\frac{\min_{i}u_i(\vec{x})}{\max_{i}u_i(\vec{x})}$ describes how fairly the utilities are divided among all agents. We choose the lowest such ratio of Nash equilibria as our fairness criterion.

%On one side, the support utility of some losing agent could be very small, in the winner utility maximization setting. This is because if this agent cannot win given the other agents' location, then this agent has no incentive to increase the support utility and may have a very low support utility.
%On the other side, losers have little power in political environment.
%Thus, to make more sense, we only consider support utility maximization setting.
In the winner utility setting, The fairness is simply 0 if there exists a losing agent, or 1 otherwise. In the support utility setting, the fairness is generally not easy to compute. However, we give a tight lower bound in such a setting.

%
%Here is an example: there are $n$ agents with the same width $\frac13$ and the distribution $f$ is,
%\begin{displaymath}
%f(x)=\left\{\begin{array}{ll}
%1.5-\epsilon & 0\leq x\leq 1/3 \\
%2\epsilon & 1/3<x\leq 2/3 \\
%1.5-\epsilon & 2/3< x\leq 1
%\end{array}\right.
%\end{displaymath} 
%There exists a Nash equilibrium that $x_1=1/6$,$x_2=5/6$ and $x_i=1/2,\forall i=3,...,n$.
%In this equilibrium, agents 1 and 2 are the winners with support utility $0.5-\epsilon/3$.
%Agents $i\in\{3,...,n\}$ are the losers with support utility $\frac{2\epsilon}{3(n-2)}$. 
%It can be arbitrarily low compared to the support utility of the winner.

%In the support utility setting, each agent seeks the largest support utility.
%We can prove the ratio between the lowest support utility and the largest support utility cannot be arbitrarily small.
We first give a lemma that bounds the ratio $\frac{\min_{i}u_i(\vec{x})}{\max_{i}u_i(\vec{x})}$ for any Nash equilibrium.
\begin{lemma}\label{lem:fair}
The utility of agent $i$ is at least $\frac{1}{2 \lceil w_j/w_i \rceil }$ fraction of the utility of agent $j$. The bound is tight. 
\end{lemma}
\begin{proof}
%Assume $c(x)$ is the congestion function without agent $i$.
In support utility maximization setting, we have
$$u_i\geq \int_s^{s+w_i}\frac{f(x)}{c(x,\vec{x}_{-i})+1}\,\mathrm{d}x\geq \frac12\int_s^{s+w_i}\frac{f(x)}{c(x,\vec{x}_{-i})}\,\mathrm{d}x,\forall s.$$ 
The idea is to split the interval $(x_j-w_j/2,x_j+w_j/2)$ into many small intervals with size $w_j$,
and apply the inequality on them:
%We let $s$ be different numbers and sum the equations up and get
\begin{align*}
\left\lceil \frac{w_j}{w_i}\right\rceil u_i &\geq \frac{1}{2} 
\sum_{k=1}^{\left\lceil \frac{w_j}{w_i}\right\rceil}
\int_{x_i-\frac{w_i}{2}+w_j\cdot(k-1)}^{x_i-\frac{w_i}{2}+w_j\cdot(k-1)}\frac{f(x)}{c(x,\vec{x}_{-i})}\,\mathrm{d}x\\
%+\int_{x_i-\frac{w_j}{2}+w_i}^{x_i-\frac{w_j}{2}+2w_i}\frac{f(x)}{c(x,\vec{x}_{-i})}\,\mathrm{d}x \\
%&&+...+\int_{x_i-\frac{w_j}{2}+(\left\lceil \frac{w_j}{w_i}\right\rceil-1)*w_i}^{x_i-\frac{w_j}{2}+\left\lceil \frac{w_j}{w_i}\right\rceil*w_i}\frac{f(x)}{c(x,\vec{x}_{-i})}\,\mathrm{d}x]\\
&\geq\frac12\int_{x_i-\frac{w_j}{2}}^{x_i-\frac{w_j}{2}+\left\lceil \frac{w_j}{w_i}\right\rceil\cdot w_i}\frac{f(x)}{c(x,\vec{x}_{-i})}\,\mathrm{d}x\\
&\geq \frac{u_j}2
\end{align*}
Consider the case distribution 
\begin{displaymath}
f(x)=\left\{\begin{array}{ll}
2/3 & x\leq 1/2 \\
4/3 & 1/2< x
\end{array}\right.
\end{displaymath} 
There are two agents with the same width $1/2$, $(x_1=0.25, x_2=0.75)$ is a Nash equilibrium.
We can see the ratio of support utility between two agent meets the bound $1/2$.
\end{proof}

Suppose agent 1 has the largest support utility, agent $n$ has the smallest support utility, 
we can easily get the ratio between the largest and smallest support utility is $\frac{1}{2 \lceil w_1/w_n \rceil }$.
%So when we consider support utility maximization setting, the support utility in Nash equilibrium is more evenly distributed among the agents.

The following theorem is immediate based on Lemma~\ref{lem:fair}.
\begin{theorem}
	The fairness in the support utility setting is at least $\frac{1}{2 \lceil w_M/w_m \rceil }$, where $w_M=\max\{w\}$ and $w_m=\min\{w\}$. The bound is tight.
\end{theorem}

\section{Price of anarchy and upper \\bound of uncovered support}
\label{sect:poa}
The price of anarchy is an important metric that measures how efficiency decrease due to agents' selfish behaviors. In particular, we define the price of anarchy as follows:
\begin{definition}[Price of anarchy]
	Given a game $G$, the price of anarchy of the game is
	\begin{gather*}
	PoA=\frac{\min_{\vec{x}\in NE(G)}\sum_{i=1}^{n}u_i(\vec{x})}{\max_{\vec{x}}\sum_{i=1}^{n}u_i(\vec{x})}
	\end{gather*}
\end{definition}

If we consider PoA in the winner utility maximization setting, the sum of the utility is always 1. There is no inefficiency.
If we consider amount of uncovered support in the winner utility maximization setting, the upper bound could reach 1, which has a poor performance.
To make the problem interesting, we mainly consider the support utility maximization setting.

First, consider the price of anarchy. %The approach is similar to the smooth method[Tim Roughgarden]
\begin{theorem}
The price of anarchy of the support utility maximization is at least $\frac{1}{2}$. The bound is tight.
\end{theorem}
\begin{proof}
Suppose the optimal location profile that maximizes the sum of support utilities
is $\vec{x}^*=(x_1^*,x_2^*,...,x_n^*)$, and the Nash equilibrium location profile is $(x_1, x_2, ..., x_n)$.
The sum of support utilities in $\vec{x}^*$ is upper bounded by adding $n$ agents with location $\vec{x}$,
\begin{align*}
\sum_{i=1}^{n}u_i(\vec{x}^*)\leq&\sum_{k=1}^{2n}u_k(\vec{x}^*,\vec{x})\\
<&\sum_{k=1}^{n}u_k(x_i^*,\vec{x}_{-i})+\sum_{k=1}^{n}u_k(x_i,\vec{x}_{-i})\\
\leq&\sum_{k=1}^{n}u_k(x_i,\vec{x}_{-i})+\sum_{k=1}^{n}u_k(x_i,\vec{x}_{-i})\\
=&2\sum_{k=1}^{n}u_k(\vec{x})
\end{align*}
%\begin{eqnarray*}
%OPT&\leq&\sum_{k=1}^{2n}u_k(\vec{x}^*,\vec{x})\\
%&<&\sum_{k=1}^{n}u_k(x_i^*,\vec{x}_{-i})+\sum_{k=1}^{n}u_k(x_i,\vec{x}_{-i})\\
%&\leq&\sum_{k=1}^{n}u_k(x_i,\vec{x}_{-i})+\sum_{k=1}^{n}u_k(x_i,\vec{x}_{-i})\\
%&=&2\sum_{k=1}^{n}u_k(\vec{x})
%\end{eqnarray*}
Thus, $PoA\geq 1/2$.
\end{proof}
Actually, when $n$ goes to infinity, the poa can be arbitrarily close to $1/2$.
Consider the example, there are $n$ agents with the same width $1/n$.
\begin{displaymath}
f(x)=\left\{\begin{array}{ll}
\frac{n^2}{2n-1}  & x\in[0,1/n] \\
\frac{n}{2n-1} & x\in(1/n,1]
\end{array}\right.
\end{displaymath} 
The optimal location profile is $(\frac{1}{2n},\frac{3}{2n},...,\frac{2n-1}{2n})$, i.e., the union of the support covers $[0,1]$ interval.
The optimal support utility is $1$.
While, consider the Nash equilibrium $(\frac{1}{2n},\frac{1}{2n},...,\frac{1}{2n})$, i.e., all the agents are located at point $\frac{1}{2n}$.
The support utility in this Nash equilibrium is $\frac{n}{2n-1}$. When $n$ goes to infinity, the PoA converges to $1/2$.

Then we can consider how many clients are not served.
%Except poa, we can also consider the upper bound of the voters without the service.
\begin{theorem}
The support of the uncovered clients is at most $\frac{1}{1+\sum_i w_i}$
\end{theorem}
\begin{proof}
We let $p$ denote the ``uncovered support'', $q$ denote ``covered support''.
Suppose $\frac{u_1}{w_1}=\min\{\frac{u_i}{w_i}\}$, i.e., agent 1 has the lowest density of the support utility.
Then the sum of all agents' support utility is at least 
$$q\geq \sum_i w_i\cdot \frac{u_1}{w_1}$$
We split the interval $[0,1]$ into pieces with size $w_1$.%%x_i-w_1/2]$ and $[x_i+w_1/2
 If we don't count agent 1, then in each small pieces, the support of the uncovered set is at most $u_1$. Otherwise, agent 1 will deviate to cover this interval. Then the sum of the uncovered support is at most $\lceil \frac{1}{w_1}\rceil \cdot u_1$
Since agent 1 has covered $u_1$, then actually the uncovered support can be limited,
$$p\leq (\left\lceil \frac{1}{w_1}\right\rceil -1)\cdot u_1\leq \frac{u_1}{w_1}$$

At last we have 
\begin{align*}
p =& \frac{p}{p+q} \\
=& \frac{1}{1+q/p}\\
\leq&\frac{1}{1+(\sum_i w_i\cdot \frac{u_1}{w_1})/(\frac{u_1}{w_1})}\\
=&\frac{1}{1+\sum_i w_i}
\end{align*}
\end{proof}

\clearpage
%\balance

\bibliographystyle{plainnat}
\bibliography{sigproc}

\begin{thebibliography}{22}
\providecommand{\natexlab}[1]{#1}
\providecommand{\url}[1]{\texttt{#1}}
\expandafter\ifx\csname urlstyle\endcsname\relax
  \providecommand{\doi}[1]{doi: #1}\else
  \providecommand{\doi}{doi: \begingroup \urlstyle{rm}\Url}\fi

\bibitem[Ben-Porat and Tennenholtz(2016)]{ben2016multi}
Omer Ben-Porat and Moshe Tennenholtz.
\newblock Multi-unit facility location games.
\newblock In \emph{International Conference on Web and Internet Economics}.
  Springer, 2016.

\bibitem[Brenner(2011)]{brenner2011location}
Steffen Brenner.
\newblock Location (hotelling) games and applications.
\newblock \emph{Wiley Encyclopedia of Operations Research and Management
  Science}, 2011.

\bibitem[Brusco et~al.(2012)Brusco, Dziubi{\'n}ski, and
  Roy]{brusco2012hotelling}
Sandro Brusco, Marcin Dziubi{\'n}ski, and Jaideep Roy.
\newblock The hotelling--downs model with runoff voting.
\newblock \emph{Games and Economic Behavior}, 74\penalty0 (2):\penalty0
  447--469, 2012.

\bibitem[Dasgupta and Maskin(1986)]{dasgupta1986existence}
Partha Dasgupta and Eric Maskin.
\newblock The existence of equilibrium in discontinuous economic games, i:
  Theory.
\newblock \emph{The Review of economic studies}, 53\penalty0 (1):\penalty0
  1--26, 1986.

\bibitem[Downs(1957)]{downs1957economic}
Anthony Downs.
\newblock An economic theory of political action in a democracy.
\newblock \emph{The journal of political economy}, pages 135--150, 1957.

\bibitem[Duggan and Fey(2005)]{duggan2005electoral}
John Duggan and Mark Fey.
\newblock Electoral competition with policy-motivated candidates.
\newblock \emph{Games and Economic Behavior}, 51\penalty0 (2):\penalty0
  490--522, 2005.

\bibitem[Eaton and Lipsey(1975)]{eaton1975principle}
B~Curtis Eaton and Richard~G Lipsey.
\newblock The principle of minimum differentiation reconsidered: Some new
  developments in the theory of spatial competition.
\newblock \emph{The Review of Economic Studies}, 42\penalty0 (1):\penalty0
  27--49, 1975.

\bibitem[Fang and Huang(2016)]{fang2016market}
Zhixuan Fang and Longbo Huang.
\newblock Market share analysis with brand effect.
\newblock In \emph{Proceedings of the 2016 International Conference on
  Autonomous Agents \& Multiagent Systems}, pages 1335--1336. International
  Foundation for Autonomous Agents and Multiagent Systems, 2016.

\bibitem[Feldman et~al.(2016)Feldman, Fiat, and
  Obraztsova]{feldman2016variations}
Michal Feldman, Amos Fiat, and Svetlana Obraztsova.
\newblock Variations on the hotelling-downs model.
\newblock In \emph{Thirtieth AAAI Conference on Artificial Intelligence}, 2016.

\bibitem[Glicksberg(1952)]{glicksberg1952further}
Irving~L Glicksberg.
\newblock A further generalization of the kakutani fixed point theorem, with
  application to nash equilibrium points.
\newblock \emph{Proceedings of the American Mathematical Society}, 3\penalty0
  (1):\penalty0 170--174, 1952.

\bibitem[Hotelling(1929)]{hotelling1929stability}
Harold Hotelling.
\newblock Stability in competition.
\newblock \emph{The Economic Journal}, 39\penalty0 (153):\penalty0 41--57,
  1929.

\bibitem[Osborne(1993)]{osborne1993candidate}
Martin~J Osborne.
\newblock Candidate positioning and entry in a political competition.
\newblock \emph{Games and Economic Behavior}, 5\penalty0 (1):\penalty0
  133--151, 1993.

\bibitem[Osborne(1995)]{osborne1995spatial}
Martin~J Osborne.
\newblock Spatial models of political competition under plurality rule: A
  survey of some explanations of the number of candidates and the positions
  they take.
\newblock \emph{Canadian Journal of economics}, pages 261--301, 1995.

\bibitem[Pinkse et~al.(2002)Pinkse, Slade, and Brett]{pinkse2002spatial}
Joris Pinkse, Margaret~E Slade, and Craig Brett.
\newblock Spatial price competition: a semiparametric approach.
\newblock \emph{Econometrica}, 70\penalty0 (3):\penalty0 1111--1153, 2002.

\bibitem[Sengupta and Sengupta(2008)]{sengupta2008hotelling}
Abhijit Sengupta and Kunal Sengupta.
\newblock A hotelling--downs model of electoral competition with the option to
  quit.
\newblock \emph{Games and Economic Behavior}, 62\penalty0 (2):\penalty0
  661--674, 2008.

\bibitem[Shaked(1975)]{shaked1975non}
A~Shaked.
\newblock Non-existence of equilibrium for the two-dimensional three-firms
  location problem.
\newblock \emph{Review of Economic Studies}, 42\penalty0 (1), 1975.

\bibitem[Shaked(1982)]{shaked1982existence}
Avner Shaked.
\newblock Existence and computation of mixed strategy nash equilibrium for
  3-firms location problem.
\newblock \emph{The Journal of Industrial Economics}, pages 93--96, 1982.

\bibitem[Thisse and Vives(1988)]{thisse1988strategic}
Jacques-Francois Thisse and Xavier Vives.
\newblock On the strategic choice of spatial price policy.
\newblock \emph{The American Economic Review}, pages 122--137, 1988.

\bibitem[Wittman(1990)]{wittman1990spatial}
Donald Wittman.
\newblock Spatial strategies when candidates have policy preferences.
\newblock \emph{Advances in the spatial theory of voting}, pages 66--98, 1990.

\bibitem[Wunder et~al.(2012)Wunder, Yaros, Kaisers, and
  Littman]{wunder2012framework}
Michael Wunder, John~Robert Yaros, Michael Kaisers, and Michael Littman.
\newblock A framework for modeling population strategies by depth of reasoning.
\newblock In \emph{Proceedings of the 11th International Conference on
  Autonomous Agents and Multiagent Systems-Volume 2}, pages 947--954.
  International Foundation for Autonomous Agents and Multiagent Systems, 2012.

\bibitem[Xefteris(2014)]{xefteris2014mixed}
Dimitrios Xefteris.
\newblock Mixed equilibria in runoff elections.
\newblock \emph{Games and Economic Behavior}, 87:\penalty0 619--623, 2014.

\bibitem[Zhang(1995)]{zhang1995price}
Z~John Zhang.
\newblock Price-matching policy and the principle of minimum differentiation.
\newblock \emph{The Journal of Industrial Economics}, pages 287--299, 1995.

\end{thebibliography}
\end{document}